\documentclass[conference]{IEEEtran}

\usepackage{mathtools}
\usepackage{amssymb}
\usepackage{amsmath}
\usepackage{algorithm}
\usepackage{algpseudocode}
\usepackage{enumerate,multirow}
\usepackage{algpseudocode}
\usepackage{tikz}
\usepackage{tikz}
\usetikzlibrary{patterns,positioning,arrows,calc}
\usepackage{siunitx}
\usepackage{multirow}
\usepackage{multicol}
\usepackage{color}
\usepackage{xcolor}
\usepackage{algorithm}
\usepackage{algpseudocode}
\usepackage{graphicx}
\usepackage{subcaption}
\algrenewcommand\algorithmicrequire{\textbf{Input:}}
\algrenewcommand\algorithmicensure{\textbf{Output:}}
\usepackage{comment}

\usepackage{geometry}
\usepackage{tikz}
\usepackage{tikz}
\usetikzlibrary{patterns,positioning,arrows,calc}\usepackage{graphicx}
\usepackage{graphicx}
\usepackage{tikz}
\usetikzlibrary{trees}
\date{}

\usepackage{calc}
\newcolumntype{M}[1]{>{\centering\arraybackslash}m{#1}}
\newcolumntype{N}{@{}m{0pt}@{}}

\usepackage{mdwmath}
\usepackage{blindtext}
\usepackage{eqparbox}
\usepackage{stfloats}
\usepackage{amsthm}
\usepackage[numbers,sort&compress]{natbib}
\usepackage{algorithm}
\usepackage{algpseudocode}
\usepackage{tikz}
\usetikzlibrary{patterns, arrows.meta, positioning}
\IEEEoverridecommandlockouts
\title{\Huge
{Univariate Bicycle Quantum LDPC Codes:\\ Explicit Logical Structure and Distance Bounds}}
\author{\IEEEauthorblockN{Sheida Rabeti and Hessam Mahdavifar} 
\IEEEauthorblockA{Department of Electrical and Computer Engineering, Northeastern University, Boston, MA 02115, USA \\ 
Email: \{rabeti.s, h.mahdavifar\}@northeastern.edu}
\thanks{This work was supported in part by NSF under Grant CCF-2415440.}
}

\theoremstyle{plain}
\newtheorem{remark}{{Remark}}
\newtheorem{example}{{Example}}
\newtheorem{theorem}{{Theorem}}
\newtheorem{lemma}[theorem]{{Lemma}}
\newtheorem{proposition}[theorem]{{Proposition}}
\newtheorem{corollary}[theorem]{{Corollary}}

\newtheorem{definition}{{Definition}}

\newcommand{\cC}{{\cal C}}

\newcommand{\cP}{{\cal P}}
 
\newcommand{\cR}{{\cal R}}
\newcommand{\cS}{{\cal S}}

\DeclareMathAlphabet{\mathbfsl}{OT1}{ppl}{b}{it} 

\newcommand{\floor}[1]{\left\lfloor #1 \right\rfloor}

\newcommand{\be}[1]{\begin{equation}\label{#1}}
\newcommand{\ee}{\end{equation}}

\renewcommand{\le}{\leqslant} 
\renewcommand{\leq}{\leqslant}
\renewcommand{\ge}{\geqslant} 
\renewcommand{\geq}{\geqslant}

\newcommand{\Cref}[1]{Co\-ro\-lla\-ry\,\ref{#1}}

\DeclareMathOperator{\wt}{wt}

\DeclareMathOperator{\supp}{supp}

\newcommand{\rs}{\mathrm{rs}}
\usepackage{hvlogos}
\usepackage{qtree}
\usepackage{tikz}
\usepackage{multirow}
\usepackage{hhline}
\usepackage{hyperref}
\usepackage{pgfplots}
\pgfplotsset{compat=1.18}
\usepackage{xcolor}

\definecolor{colorPair1}{HTML}{1F77B4}
\definecolor{colorPair2}{HTML}{D62728}
\definecolor{colorPair3}{HTML}{2CA02C}

\begin{document}

\vspace{10mm}
\maketitle

\begin{abstract}
We introduce univariate bicycle (UB) codes, a structured subclass of generalized bicycle (GB) quantum low-density parity-check (LDPC) codes obtained via a Frobenius relation. This construction reduces the code design space from a two-polynomial search in GB codes to a single-polynomial search, while preserving sparsity. We provide an explicit algebraic characterization of the logical coset spaces by constructing a basis for the logical quotient space, yielding a complete parametrization of logical operators. Leveraging this structure, we derive upper bounds on the minimum distance by relating structured logical representatives to cycle-density properties of associated circulant matrices. Finally, simulation results for short- to medium-length UB codes (block lengths ranging from a few hundred to approximately $10^3$) demonstrate competitive performance relative to existing GB and bivariate bicycle (BB) codes despite the additional algebraic restriction.
\end{abstract}

\section{Introduction}
Quantum error-correcting codes are essential for fault-tolerant quantum computation. Among existing constructions, quantum low-density parity-check (QLDPC) codes have received significant attention in recent years due to their sparse parity-check matrices, which enable scalable implementations with bounded-weight stabilizers. Building on classical LDPC codes introduced by Gallager \cite{gallager1962low}, recent advances in QLDPC code design have focused on developing structured constructions that balance sparsity, distance, and implementability \cite{7336474,Breuckmann_2021,vasic2025quantumlowdensityparitycheckcodes}. In particular, the design of algebraically structured code families plays a key role in enabling tractable analysis and efficient exploration of the code space.

Within this class, generalized bicycle (GB) codes \cite{PhysRevA.88.012311, Panteleev_2021}, which extend the original bicycle construction \cite{MacKay_2004}, form a flexible family of Calderbank–Shor–Steane (CSS) codes \cite{PhysRevA.54.1098, PhysRevA.54.4741} with desirable structural and performance properties. An important subclass is given by bivariate bicycle (BB) codes, which have recently attracted significant attention due to their performance and compatibility with near-term quantum hardware \cite{Bravyi_2024}. Coprime BB codes further refine this framework by imposing additional algebraic constraints on the defining polynomials, leading to improved code parameters \cite{wang2024coprime, postema2025existencecharacterisationbivariatebicycle}. However, these constructions typically require searching over pairs of polynomials, resulting in a large design space and substantial computational overhead for code search.

Motivated by the need for more structured and tractable constructions, we introduce a new subclass called \emph{univariate bicycle (UB) codes}, in which the defining polynomials are coupled through a Frobenius relation. This restriction reduces the design space from a two-polynomial search in GB/BB constructions to a single-polynomial search, while preserving the LDPC structure. More importantly, this structure enables an explicit algebraic characterization of the logical operators. Although general methods \cite{gottesman1997stabilizercodesquantumerror, Wilde_2009} provide procedures for determining logical operators, they typically do not result in compact or structured representations. In contrast, we construct an explicit basis for the logical coset spaces, yielding a complete parameterization of the logical operators of UB codes.

Leveraging this explicit structure, we derive upper bounds on the minimum distance by restricting attention to structured families of logical representatives and relating their weights to cycle-density properties of the associated circulant matrices. This is particularly significant since computing the minimum distance of stabilizer codes is, in general, NP-hard \cite{10154053,641542,7891006}. Beyond distance analysis, the explicit logical basis also provides a basis for systematic logical gate design, which is often difficult for general QLDPC codes. Finally, we present simulation results for short- to medium-length UB codes under standard decoding algorithms, demonstrating their competitive performance and highlighting the potential of the proposed construction in practical settings.

\section{Preliminaries}
\subsection{Quantum Stabilizer and CSS Codes}
An $[[N, k, d]]$ quantum stabilizer code is a $2^k$-dimensional subspace $\mathcal{C} \subseteq (\mathbb{C}_2)^{\otimes N}$ with an Abelian stabilizer group 
$\mathcal{S}$:
\begin{equation} \label{eq:StabilizerCode} \mathcal{C} = \{|\psi\rangle \in (\mathbb{C}_2)^{\otimes N} \colon s |\psi\rangle= |\psi\rangle, \forall s \in \mathcal{S} \}. \end{equation}
Each generator $g \in \mathcal{S}$ acts as a parity-check constraint. The \textit{minimum distance} $d$ of a stabilizer code is the minimum weight of some Pauli operator $P \in \cP_N$, where $\cP_N$ is the $N$-qubit Pauli group, commuting with elements in $\cS$ such that $P \notin \cS$. A Calderbank--Shor--Steane (CSS) code is a stabilizer code with a parity-check matrix of the form
$
H = \begin{bmatrix} H_X & 0 \\ 0 & H_Z \end{bmatrix},
$
where $H_X,H_Z$ are classical binary parity-check matrices satisfying $H_X H_Z^T = 0$. Such codes can correct Pauli-$X$ and Pauli-$Z$ errors independently using $H_Z$ and $H_X$, respectively. In this work, we focus on QLDPC CSS codes, where both $H_X$ and $H_Z$ are sparse.
For CSS codes, the minimum distance admits an equivalent characterization in terms of the underlying classical codes. Specifically, the $Z$-distance is defined as
\[
d_Z = \min \{ \mathrm{wt}(v) : v \in \ker(H_Z)\setminus \operatorname{rs}(H_X) \},
\]
and the $X$-distance is also defined in the same fashion. 
The minimum distance of the CSS code is given by $d=\min(d_X,d_Z)$. Vectors in $\ker(H_Z)$ correspond to Pauli-$Z$ operators that commute with all stabilizers, while vectors in $\operatorname{rs}(H_X)$ correspond to stabilizer generators themselves. Hence, elements of $\ker(H_Z)\setminus \operatorname{rs}(H_X)$ represent nontrivial logical $Z$ operators. Similarly, $\ker(H_X)\setminus \operatorname{rs}(H_Z)$ characterizes nontrivial logical $X$ operators. Throughout, we use $\ker(\cdot)$ and $\operatorname{rs}(\cdot)$ to denote the kernel (null-space) and row space of a matrix, respectively.
\subsection{Generalized Bicycle (GB) Codes}
A code is called \emph{cyclic} if it is closed under cyclic shifts. For any cyclic code $\cC$, we can associate a one-to-one mapping between $\mathbb{F}^n$ and $\cR_n \triangleq \mathbb{F}[x] / (x^n-1)$ by mapping $c = (c_{0}, c_1, \ldots, c_{n-1}) \in \mathbb{F}^n$ to $c(x) \triangleq c_0 + c_1x + ... + c_{n-1}x^{n-1}$.
Then every cyclic code forms an ideal in $\mathcal{R}_n$ generated by a monic polynomial 
$g(x)$ where $g(x)\,|\,x^n - 1$, with check polynomial $h(x) = (x^n - 1)/g(x)$. 
Let $P$ be the $n \times n$ cyclic permutation matrix. Then, both generator and parity-check matrices can be expressed as circulant matrices where a circulant matrix $A$ corresponding to the polynomial $a(x) = a_0 + a_1x+ ... + a_{n-1}x^{n-1}$ is defined as $A = \mathrm{Circ}(a(x))$.

Generalized bicycle (GB) codes \cite{MacKay_2004,PhysRevA.88.012311}
are CSS codes constructed from two circulant matrices
$A=\mathrm{Circ}(a(x))$ and $B=\mathrm{Circ}(b(x))$, where
$a(x),b(x)\in\mathbb F_2[x]$. The corresponding parity-check matrices are
$
H_X=[A,B],\ 
H_Z=[B^T,A^T].
$
Note that \(B^T\) and \(A^T\) are generated, up to a cyclic shift, by the reciprocal polynomials
$b^*(x)=x^{\deg(b)}b(x^{-1})$ and
$a^*(x)=x^{\deg(a)}a(x^{-1})$,
respectively. The CSS orthogonality condition holds since
$
H_X H_Z^{\mathsf T}=AB+BA=0,
$
where circulant matrices commute. Moreover, $H_X$ and $H_Z$
are permutation-equivalent, implying that the corresponding
classical codes have identical distances. Hence,
$
d_X=d_Z=d.
$

\begin{proposition}\cite{Panteleev_2021}
\label{prop1}
The dimension $k$ of the generalized bicycle code $[[N = 2n, k, d]]$ defined by 
$a(x), b(x) \in \mathbb{F}_2[x]$ is given by
$
    k = 2 \deg \tilde{h}(x),
$
where $
\tilde{h}(x) \triangleq \gcd\big(a(x), b(x), x^n - 1\big).
$
\end{proposition}
Let $\mathrm{GB}(a,b)$ denote the GB code defined by the circulant matrices corresponding to the polynomials $a(x), b(x) \in \mathbb{F}_2[x]$.

\section{Univariate Bicycle (UB) Code}

In this section, we introduce univariate bicycle (UB) codes that form a structured subclass of generalized bicycle (GB) codes. In particular, they are obtained by restricting the two defining polynomials in GB codes to a single generator.

Let $\mathcal{R}_n \triangleq \mathbb{F}_2[x]/(x^n - 1)$, and let $a(x) \in \mathbb{F}_2[x]$ be a polynomial such that $r = \deg\big(\gcd(a(x), x^n - 1)\big) > 0$.
For an integer $\ell \geq 1$, define
\begin{equation}
b(x) \triangleq a(x)^{2^\ell} \;\; \text{in } \mathcal{R}_n.
\end{equation}
Let $A = \mathrm{Circ}(a(x))$ and $B = \mathrm{Circ}(b(x))$ denote the corresponding circulant matrices. The resulting CSS code $[[N = 2n, k, d]]$ with parity-check matrices
\begin{equation}
H_X = [A,\, B], \qquad H_Z = [B^{\mathsf{T}},\, A^{\mathsf{T}}],
\end{equation}
is called a \emph{univariate bicycle (UB) code}, denoted by $\mathrm{UB}(a(x), \ell)$. Equivalently, this construction can be expressed as
$\mathrm{UB}(a(x), \ell) = \mathrm{GB}\big(a(x),\, a(x)^{t}\big)$ with $t = 2^\ell$.

This construction satisfies the CSS orthogonality condition $H_X H_Z^{\mathsf{T}} = 0$, since circulant matrices commute. The following lemma describes how the Frobenius relation
$b(x)=a(x)^{2^\ell}$ affects the support and Hamming weight in
$\mathcal R_n$. 

\begin{lemma}
\label{weight-lemma}
Let $a(x)=\sum_{i=0}^{n-1} a_i x^i \in \mathcal R_n$ and let
$b(x)=a(x)^{2^\ell}$ for some $\ell \ge 1$. Then
$
    b(x)=a(x^{2^\ell})=\sum_{i=0}^{n-1} a_i x^{2^\ell i}
    \quad \text{in } \mathcal R_n .
$
Consequently, $\wt(b)\le \wt(a)$. Moreover, equality holds if
$\gcd(n,2^\ell)=1$; in particular, equality holds for odd $n$.
\end{lemma}

\begin{proof}
The identity follows from the Frobenius map over $\mathbb F_2$.
In $\mathcal R_n$, exponents are reduced modulo $n$. Thus, distinct
nonzero terms of $a(x)$ may collide after the map
$i \mapsto 2^\ell i \pmod n$, implying the Hamming weight cannot increase.
If $\gcd(n,2^\ell)=1$, this map is a permutation of $\mathbb Z_n$ and hence, no collisions occur and the weight is preserved.
\end{proof}

Compared to general GB codes, where $a(x)$ and $b(x)$ are chosen
independently, the UB construction reduces the search space by setting
$b(x)$ as a function of $a(x)$ while preserving the low-weight structure. In
particular, if $\widetilde h(x)=\gcd(a(x),x^n-1)$ with
$\deg(\widetilde h)=r$, then Proposition \ref{prop1} gives
$k=2r$.


\section{Logical Coset Structure and Distance Bounds for UB Codes}
\label{sec:construct}

We now derive explicit algebraic descriptions of the logical quotient
spaces of UB codes and use them to obtain structured distance bounds.
In this section, we restrict attention to the divisor case
$a(x)\mid(x^n-1)$, which enables tractable ideal-theoretic descriptions of the logical operators through the induced annihilator structure. Now, consider the
$\mathrm{UB}(a(x),\ell)$ code with
\vspace{-2mm}
\[
t\triangleq 2^\ell,
\qquad
h(x)\triangleq \frac{x^n-1}{a(x)},
\qquad
r\triangleq \deg(a(x)).
\]

Define the logical quotient spaces
\[
\mathcal L_Z
\triangleq
\ker(H_Z)/\rs(H_X),
\qquad
\mathcal L_X
\triangleq
\ker(H_X)/\rs(H_Z).
\]
For $\mathbf z\in\ker(H_Z)$ and $\mathbf x\in\ker(H_X)$, denote the
corresponding cosets in $\mathcal L_Z$ and $\mathcal L_X$ by
$[\mathbf z]\triangleq \mathbf z+\rs(H_X)$ and
$[\mathbf x]\triangleq \mathbf x+\rs(H_Z)$, respectively.

\begin{theorem}[Logical basis of UB codes]
\label{thm:logical-basis-ub}
Let $a^*$, $b^*$, and $h^*$ denote the reciprocal
polynomials of $a(x)$, $b(x)$, and $h(x)$, respectively. For
$i,j\in\{0,\dots,r-1\}$, define
\[
Z_i^{(1)}
\triangleq
\big(x^i,(a^*)^{t-1}x^i\big),
\qquad
Z_j^{(2)}
\triangleq
\big(0,h^*x^j\big).
\]

Then the set
\[
\mathcal B_Z
\triangleq
\Big\{
[Z_i^{(1)}]
\Big\}_{i=0}^{r-1}
\cup
\Big\{
[Z_j^{(2)}]
\Big\}_{j=0}^{r-1}
\]
forms an $\mathbb F_2$-basis of $\mathcal L_Z$. Moreover, for $i,j\in\{0,\dots,r-1\}$, define
\[
X_i^{(1)}
\triangleq
\big(a^{t-1}x^i,x^i\big),
\qquad
X_j^{(2)}
\triangleq
\big(hx^j,0\big).
\]

Then the set
\[
\mathcal B_X
\triangleq
\Big\{
[X_i^{(1)}]
\Big\}_{i=0}^{r-1}
\cup
\Big\{
[X_j^{(2)}]
\Big\}_{j=0}^{r-1}
\]
forms an $\mathbb F_2$-basis of $\mathcal L_X$.

\end{theorem}
\begin{proof}
We prove the claim for $\mathcal L_Z$; the statement for $\mathcal L_X$
follows analogously.

Since $b=a^t$ and $b^*(x)=x^{\deg(b)}b(x^{-1})$,  we have $b^*=(a^*)^t$. Moreover, from
$ah=x^n-1$, it follows that
$a^*h^*=x^n-1=0$ in $\mathcal R_n$. Since
$a^*\mid(x^n-1)$, the set of all polynomials
$y\in\mathcal R_n$ satisfying $a^*y=0$ is precisely
the ideal
$(h^*)=\{h^*q:q\in\mathcal R_n\}$. Hence
\begin{align*}
\ker(H_Z)
&=
\{(u,v)\in\mathcal R_n^2:\ b^*u+a^*v=0\}
\\
&=
\{(p,(a^*)^{t-1}p+h^*q):p,q\in\mathcal R_n\}.
\end{align*}
Therefore $Z_i^{(1)},Z_j^{(2)}\in\ker(H_Z)$. Moreover,
\[
\rs(H_X)=\{(as,a^ts)\in\mathcal R_n^2:s\in\mathcal R_n\}.
\]

\textit{Spanning:}
Let $(u,v)\in\ker(H_Z)$. Then
$(u,v)=(p,(a^*)^{t-1}p+h^*q)$ for some
$p,q\in\mathcal R_n$. Since every element modulo $(a)$ admits a unique
representative of degree smaller than $r=\deg(a)$, write
$p=p_0+as$ with $\deg(p_0)<r$. Modulo $\rs(H_X)$,
\begin{equation}
\label{eq:mod_rs}
    (u,v)
\sim
(p_0,(a^*)^{t-1}p+h^*q+a^ts).
\end{equation}
Since $\rs(H_X)\subseteq\ker(H_Z)$, the vector in \eqref{eq:mod_rs}
lies in $\ker(H_Z)$ and thus equals
$(p_0,(a^*)^{t-1}p_0+h^*\widetilde q)$
for some $\widetilde q\in\mathcal R_n$. Similarly, every element modulo
$(a^*)$ admits a unique representative of degree smaller than $r$, so
write $\widetilde q=q_0+a^*z$ with $\deg(q_0)<r$. Using
$a^*h^*=0$, we obtain
$[(u,v)]
=
[(p_0,(a^*)^{t-1}p_0+h^*q_0)].$
Expanding
$p_0=\sum_{i=0}^{r-1}\alpha_i x^i$
and
$q_0=\sum_{j=0}^{r-1}\beta_j x^j$
gives
\[
[(u,v)]
=
\sum_{i=0}^{r-1}\alpha_i[Z_i^{(1)}]
+
\sum_{j=0}^{r-1}\beta_j[Z_j^{(2)}].
\]
Thus $\mathcal B_Z$ spans $\mathcal L_Z$.

\textit{Linear independence:}
Suppose
$(p_0,(a^*)^{t-1}p_0+h^*q_0)
=
(as,a^ts),$
where $\deg(p_0),\deg(q_0)<r$. From the first coordinate,
$p_0\equiv as \pmod{x^n-1}$, so $a\mid p_0$ in $\mathbb F_2[x]$.
Since $\deg(p_0)<r=\deg(a)$, it follows that $p_0=0$.
Then $as=0$ in $\mathcal R_n$, implying that $s$ is a multiple of $h$,
and therefore $a^ts=0$. The second coordinate gives
$h^*q_0=0$. Since $h^*a^*=0$ in $\mathcal R_n$ and
$h^*\mid(x^n-1)$, it follows that $q_0$ is a multiple of $a^*$.
Because $\deg(q_0)<r=\deg(a^*)$, we conclude that $q_0=0$.
Hence all coefficients vanish, and $\mathcal B_Z$ is linearly independent.
Therefore $\mathcal B_Z$ is an $\mathbb F_2$-basis of $\mathcal L_Z$.
\end{proof}

\begin{remark}[Parametrization of logical classes]
\label{remark:parametrization}

For $\alpha,\beta\in\mathbb F_2^r$, define
\[
\Lambda_Z(\alpha,\beta)
\triangleq
\sum_{i=0}^{r-1}\alpha_i Z_i^{(1)}
+
\sum_{j=0}^{r-1}\beta_j Z_j^{(2)}.
\]
Since $\mathcal B_Z$ forms a basis of $\mathcal L_Z$, the map
$(\alpha,\beta)
\longmapsto
[\Lambda_Z(\alpha,\beta)]$
defines a bijection from $\mathbb F_2^{2r}$ onto $\mathcal L_Z$.
Consequently, every coset in $\mathcal L_Z$ admits a unique representative of the form
$\Lambda_Z(\alpha,\beta).$

In particular,
$[\Lambda_Z(\alpha,\beta)]
=
\Lambda_Z(\alpha,\beta)+\rs(H_X),$
and the corresponding logical class is trivial if and only if
$(\alpha,\beta)=(0,0).$
Equivalently, a logical class is nontrivial precisely when
$(\alpha,\beta)\neq(0,0).$
Similarly, define
\[
\Lambda_X(\alpha,\beta)
\triangleq
\sum_{i=0}^{r-1}\alpha_i X_i^{(1)}
+
\sum_{j=0}^{r-1}\beta_j X_j^{(2)}.
\]
Then the map
$(\alpha,\beta)
\longmapsto
[\Lambda_X(\alpha,\beta)]$
defines a bijection from $\mathbb F_2^{2r}$ onto $\mathcal L_X$, and every coset in $\mathcal L_X$ admits a unique representative of the form
$\Lambda_X(\alpha,\beta).$
\end{remark}
Using the distance characterization of CSS codes from Section II, together
with Remark \ref{remark:parametrization}, every nontrivial logical class in
$\mathcal L_Z$ admits a unique representative of the form
$\Lambda_Z(\alpha,\beta)$ with
$(\alpha,\beta)\neq(0,0)$. Hence
\[
d_Z
=
\min_{\substack{
(\alpha,\beta)\neq(0,0)\\
s\in\rs(H_X)
}}
\wt\big(
\Lambda_Z(\alpha,\beta)+s
\big).
\]
Similarly,
\[
d_X
=
\min_{\substack{
(\alpha,\beta)\neq(0,0)\\
s\in\rs(H_Z)
}}
\wt\big(
\Lambda_X(\alpha,\beta)+s
\big).
\]

Restricting attention to logical representatives
$\Lambda_X(\alpha,\beta)$ and $\Lambda_Z(\alpha,\beta)$ with
$(\alpha,\beta)\neq(0,0)$ and
$\wt(\alpha)+\wt(\beta)\le q$
yields the following bounds on the minimum distance.
\begin{corollary}[Weight-$q$ distance upper bounds]
\label{cor:low-order-bound}
Let $f=a^{t-1}$. For $q\ge1$, define
\[
U_q^X\triangleq
\min_{\substack{
(\alpha,\beta)\neq(0,0)\\
\wt(\alpha)+\wt(\beta)\le q
}}
\wt((\alpha f+\beta h,\alpha)),
\]
where $\alpha,\beta$ are polynomials of degree less than $r$. Define
$U_q^Z$ analogously from $\Lambda_Z(\alpha,\beta)$. Then
\[
d\le \min\{U_q^X,U_q^Z\}.
\]
\end{corollary}
The quantities $U_q^X$ and $U_q^Z$ can be bounded through overlap
patterns among shifted copies of $f=a^{t-1}$ and $h$. To quantify these
overlaps, let $\mathrm{Circ}(g)$ denote the circulant matrix generated by
$g(x)\in\mathcal R_n$, and let
\[
\mathcal C_s(g)
\]
denote the submatrix formed by the first $s$ rows of
$\mathrm{Circ}(g)$.

\begin{definition}[Maximum induced $2c$-cycle density]
Let $H$ be a binary parity-check matrix with $m$ rows, and let
$\mathcal G(H)$ denote its Tanner graph. For a subset
$\mathcal S\subseteq[m]$ with $|\mathcal S|=c$, let
\[
N^{\mathrm{ind}}_{2c}(\mathcal S)
\]
denote the number of induced, i.e., chordless, cycles of length $2c$
in $\mathcal G(H)$
whose set of check nodes is exactly $\mathcal S$. Define
\[
\rho^{\mathrm{ind}}_{2c}(H)
\triangleq
\max_{\substack{
\mathcal S\subseteq[m]\\
|\mathcal S|=c
}}
N^{\mathrm{ind}}_{2c}(\mathcal S).
\]
\end{definition}

The following theorem derives bounds for $U_q^X$ from
$\mathcal C_r(f),\mathcal C_r(h)$ and for $U_q^Z$ from
$\mathcal C_r(f^*),\mathcal C_r(h^*)$, through the induced $4$- and
$6$-cycle densities of these submatrices.
\begin{theorem}[Distance bounds from induced-cycle densities]
\label{thm:cycle-induced-bounds}
Let $f=a^{t-1}$. Then
\begin{equation}
    U_1^X
\le
B_1^X
\triangleq
\min\{\wt(f)+1,\wt(h)\}.
\label{eq:B1}
\end{equation}
Moreover,
\begin{align}
U_2^X
\le
B_2^X
\triangleq
\min\Big\{
&\,2\wt(f)+1- 
\sqrt{1+8\rho_4^{\mathrm{ind}}(\mathcal C_r(f))},
\nonumber\\
&
2\wt(h)
-1 - \sqrt{1+8\rho_4^{\mathrm{ind}}(\mathcal C_r(h))}
\Big\},
\label{eq:B2}
\end{align}
and
\begin{align}
U_3^X
\le
B_3^X
\triangleq
\min\Big\{
&\,3\wt(f)+3
-6\big(\rho^{\rm ind}_6(\mathcal C_r(f))\big)^{1/3},
\nonumber\\
&
3\wt(h)
-6\big(\rho^{\rm ind}_6(\mathcal C_r(h))\big)^{1/3}
\Big\}.
\label{eq:B3}
\end{align}

The bounds $B_q^Z$ are obtained similarly from
$\Lambda_Z(\alpha,\beta)$. Hence,
\[
d\le
\min_{q\in\{1,2,3\}}
\{B_q^X,B_q^Z\}.
\]
\end{theorem}

\begin{proof}
By Corollary \ref{cor:low-order-bound}, any nonzero representative $\Lambda_X(\alpha,\beta)$ with $\wt(\alpha)+\wt(\beta)\le q$ provides an upper bound on $d$. For any $q\ge1$, restricting to $\beta=0$ or $\alpha=0$ gives
\[
U_q^X
\le
\min\!\left\{
\min_{\substack{\alpha\neq0\\\wt(\alpha)\le q}}\wt(\alpha f,\alpha),\;
\min_{\substack{\beta\neq0\\\wt(\beta)\le q}}\wt(\beta h,0)
\right\}.
\]
Weight $q=1$: Taking $(\alpha,\beta)=(x^i,0)$ and $(0,x^j)$ directly gives $U_1^X=\min\{\wt(f)+1,\wt(h)\}$.

Weight $q=2$: Consider two rows $fx^{i_1}$ and $fx^{i_2}$ of
$\mathcal C_r(f)$, and let
$\mathcal N_{12}
\triangleq
\supp(fx^{i_1})\cap\supp(fx^{i_2})$
denote the set of columns containing ones in both rows. Then the representative
$(fx^{i_1}+fx^{i_2},x^{i_1}+x^{i_2})$ has weight
$2\wt(f)+2-2|\mathcal N_{12}|$. Moreover, every pair of columns in
$\mathcal N_{12}$ induces a chordless $4$-cycle on the two selected rows,
so
$N_4^{\mathrm{ind}}=\binom{|\mathcal N_{12}|}{2}.$
Thus
$|\mathcal N_{12}|=\frac{1+\sqrt{1+8N_4^{\mathrm{ind}}}}{2}$.
Maximizing over row pairs yields \ref{eq:B2}.
The same argument on the $h$-family gives the second term; taking the minimum gives $B_2^X$.

Weight $q=3$: Consider three rows $fx^{i_1},fx^{i_2},fx^{i_3}$ of
$\mathcal C_r(f)$. For $\{a,b,c\}=\{1,2,3\}$, define $\mathcal N_{ab}$ as the set of weight-$2$ columns supported
exactly on rows $i_a$ and $i_b$ but not $i_c$. Then
\[
\wt(fx^{i_1}+fx^{i_2}+fx^{i_3})
\le
3\wt(f)-2\big(
|\mathcal N_{12}|+|\mathcal N_{13}|+|\mathcal N_{23}|
\big).
\]
Choosing one column from each of
$\mathcal N_{12},\mathcal N_{13},\mathcal N_{23}$ gives a distinct
chordless $6$-cycle on the three selected rows; hence
$N_6^{\mathrm{ind}}
=
|\mathcal N_{12}|\,|\mathcal N_{13}|\,|\mathcal N_{23}|.$
By AM--GM,
$|\mathcal N_{12}|+|\mathcal N_{13}|+|\mathcal N_{23}|
\ge
3\big(N_6^{\mathrm{ind}}\big)^{1/3}.$
Therefore, maximizing over all triples of rows gives the first term in
\eqref{eq:B3}. The same argument applied to the $h$-family gives the
second term.
\end{proof}
\begin{corollary}[Weight-1 distance bound]
\label{cor:first-order-ub}
For any $\mathrm{UB}(a(x),\ell)$ code in the divisor case,$
d\le \min\{\wt(a)^\ell+1,\wt(h)\}$.
\end{corollary}

\begin{proof}
By Theorem \ref{thm:cycle-induced-bounds},
$
d \le B_1^X = \min\{\wt(a^{t-1})+1,\wt(h)\}.
$
Since $t-1=\sum_{i=0}^{\ell-1}2^i$, we have
$
a^{t-1}=\prod_{i=0}^{\ell-1} a^{2^i}.
$
By Lemma \ref{weight-lemma}, $\wt(a^{2^i})\le \wt(a)$ for each $i$.
Using $\wt(pq)\le \wt(p)\wt(q)$ in $\mathcal R_n$, we obtain
\[
\wt(a^{t-1})
\le
\prod_{i=0}^{\ell-1}\wt(a^{2^i})
\le
\wt(a)^\ell.
\]
The claim follows.
\end{proof}
\begin{example}[Tight low-order bounds]
\label{example:tightness}
Consider the code $[[42,8,5]]=\mathrm{UB}(1+x+x^2+x^4,1)$ over
$\mathcal R_{21}=\mathbb F_2[x]/(x^{21}-1)$, where
$f(x)=a(x)^{2^1-1}=1+x+x^2+x^4$ and
$h(x)=1+x+x^3+x^7+x^8+x^{10}+x^{14}+x^{15}+x^{17}$.
Here $\wt(a)=\wt(f)=4$ and $\wt(h)=9$.

The weight-1,2, and 3 searches give
$U_1^X= U_2^X = U_3^X = 5.$
From the first-order bound \ref{eq:B1},
$B_1^X
=
\min\{\wt(f)+1,\wt(h)\}
=
5.$
Moreover,
$\rho^{\mathrm{ind}}_4(\mathcal C_4(f))=1$ and
$\rho^{\mathrm{ind}}_4(\mathcal C_4(h))=3$, yielding
\[
B_2^X
=
\min\{9 - \sqrt9,\;17 - \sqrt{25}\}
=
6.
\]
Similarly,
$\rho^{\mathrm{ind}}_6(\mathcal C_4(f))=4$ and
$\rho^{\mathrm{ind}}_6(\mathcal C_4(h))=27$, giving
\[
B_3^X
=
\min\{15-6\cdot4^{1/3},\;27-18\}
=
5.475.
\]

Hence, $B_1^X$ and $\lfloor B_3^X\rfloor$ both recover the exact distance $d=5$.
\end{example}
\begin{remark}
For the code in Example \ref{example:tightness}, both the weight-1 and
weight-3 bounds recover the exact distance $d=5$. In particular,
\[
B_1^X=5,
\qquad
\lfloor B_3^X\rfloor=5.
\]

The first-order bound
$d_{GB}\le \wt(f)+1$
was previously observed for general GB codes
in \cite{dastbasteh2025generalized,wang2022distance}. However,
Corollary \ref{cor:low-order-bound} and
Theorem \ref{thm:cycle-induced-bounds} yield substantially tighter
code-dependent bounds for specific UB instances.

For example, the following code over $\mathcal R_{30}$
\[
[[60,8,5]]
=
\mathrm{UB}(1+x+x^3+x^4,5)
\]
satisfies
$
B_1^X=15,
B_2^X=6,$
while
$
d = \floor{B_3^X}=\floor{5.0897}.
$
\end{remark}

\noindent
{\textbf{Complexity Analysis.}
In the UB construction, the polynomial $b(x)$ is determined
directly from $a(x)$ through the Frobenius relation
$b(x)=a(x)^{2^\ell}$. Hence, the search is performed only over $a(x)$.
By Lemma \ref{weight-lemma},
$
\wt(b)\le \wt(a).
$
Hence, for a target stabilizer weight bound $w$, it suffices to search
over polynomials satisfying
$
\wt(a)\le \frac{w}{2}
$.
Therefore, for fixed $w$, the search space scales as
$\mathcal O(n^{w/2})$. In contrast, general GB constructions require a joint search over both
$a(x)$ and $b(x)$ under the constraint
$\wt(a)+\wt(b)\le w$, leading to a search space scaling with
$\mathcal O(n^w)$.
Thus, the UB restriction substantially reduces the search space while
preserving the low-weight structure.
}
\section{Numerical Results}
\begin{table}[t]
\centering

\label{tab:poly-params}
\renewcommand{\arraystretch}{1.0}
\setlength{\tabcolsep}{3.4pt}
\small
\begin{tabular}{|c|>{\centering\arraybackslash}m{1.2cm}|c|c|c|}
\hline
$a(x)$ & $\ell$ & $[[N,k,d]]$ & $R = \frac{k}{N}$ & $w$ \\
\hline
$x^{7}+x^{4}+x+1$ & 3 & $[[124,14,11]]$ & 0.113 & 8 \\
\hline
$x^{10}+x^{9}+x^{2}+1$ & 4 & $[[146,20,8]]$ & 0.137 & 8 \\
\hline
$x^{12}+x^{10}+x^{9}+1$ & 5 & $[[178,24,13]]$ & 0.135 & 8 \\
\hline
$x^{18}+x^{8}+x^{4} + 1$ & 6 & $[[204, 36, 8]]$ & 0.176 & 8 \\
\hline
$x^{13}+x^{5}+x + 1$ & 5 & $[[234, 26, 14]]$ & 0.111 & 8 \\
\hline
$x^{6}+x^{5}+1$ & 3 & $[[252, 12, 14]]$ & 0.048 & 6 \\
\hline
$x^{7}+x^{4}+1$ & 3 & $[[254, 14, 14]]$ & 0.055 & 6 \\
\hline
$x^{7}+x^{2}+1$ & 5 & $[[372, 14, 12]]$ & 0.038 & 6 \\
\hline
$x^{6}+x^{1}+1$ & 9 & $[[378, 12, 22]]$ & 0.032 & 6 \\
\hline
$x^{9}+x^{8}+1$ & 7 & $[[730, 18, 20]]$ & 0.025 & 6 \\
\hline
$x^{28}+x^{10}+x+1$ & 4 & $[[1022, 56, 21]]$ & 0.057 & 8 \\
\hline
\end{tabular}
\caption{Selected UB codes with $b(x) = a(x)^t$, $t = 2^\ell$, $N=2n$.}
\label{tab:poly-params}
\vspace{-4.8mm}
\end{table}

\definecolor{colorPair1}{HTML}{1F77B4}   
\definecolor{colorPair2}{HTML}{D62728}   
\definecolor{colorPair3}{HTML}{2CA02C}   

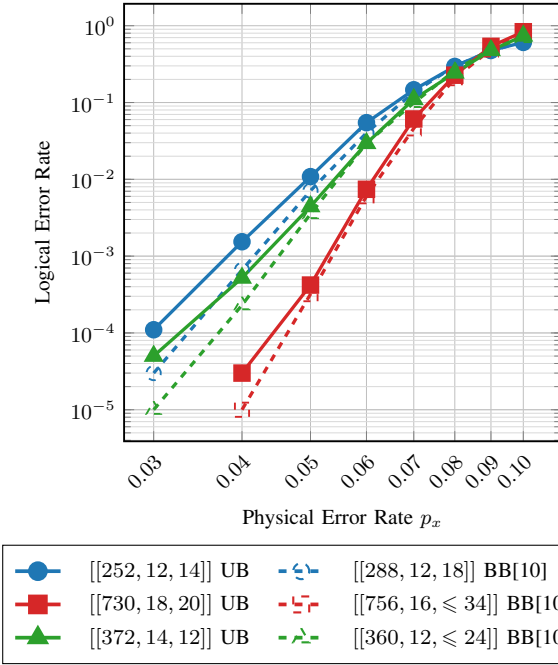
\begin{figure}[t]
    \centering
    \resizebox{0.79\columnwidth}{!}{%
    \begin{tikzpicture}
    \begin{loglogaxis}[
        scale only axis,
        width  = 6.5cm,
        height = 6.5cm,
        xlabel = {Physical Error Rate $p_x$},
        ylabel = {Logical Error Rate},
        xlabel style     = {font=\small},
        ylabel style     = {font=\small},
        tick label style = {font=\small},
        xmin = 0.028, xmax = 0.11,
        ymin = 5e-6,  ymax = 1.5,
        xtick       = {0.03, 0.04, 0.05, 0.06, 0.07, 0.08, 0.09, 0.10},
        xticklabels = {$0.03$, $0.04$, $0.05$, $0.06$,
                       $0.07$, $0.08$, $0.09$, $0.10$},
        xticklabel style = {rotate=45, anchor=north east, font=\small},
        xtick align      = inside,
        ytick align      = inside,
        enlarge x limits = 0.02,
        enlarge y limits = 0.02,
        clip             = true,
        grid             = both,
        grid style       = {line width=0.4pt, black!25},
        minor grid style = {line width=0.3pt, black!15},
        minor tick num   = 1,
        axis line style  = {line width=0.8pt},
        legend style = {
    font         = \footnotesize,
    fill         = white,
    draw         = black,
    cells        = {anchor=west},
    column sep   = 8pt,
    row sep      = 1pt,
    inner sep    = 4pt,
    legend columns = 2,
},
legend to name    = ferlegendBB,
legend cell align = left,
    ]

    \addplot[color=colorPair1, mark=*, mark size=3pt,
             line width=1.4pt, solid] coordinates {
        (0.03, 0.00011)(0.04, 0.00154)(0.05, 0.01084)(0.06, 0.05467)
        (0.07, 0.14658)(0.08, 0.29492)(0.09, 0.47678)(0.10, 0.60643)
    };
    \addlegendentry{$[[252,12, 14]]$ UB}

    \addplot[color=colorPair1, mark=o, mark size=3pt,
             mark options={fill=white, draw=colorPair1, line width=0.8pt},
             line width=1.4pt, dashed] coordinates {
        (0.03, 0.00003)(0.04, 0.00066)(0.05, 0.00711)(0.06, 0.04027)
        (0.07, 0.13579)(0.08, 0.29883)(0.09, 0.46749)(0.10, 0.62651)
    };
    \addlegendentry{$[[288,12,18]]$ BB\cite{Bravyi_2024}}

    \addplot[color=colorPair2, mark=square*, mark size=3pt,
             line width=1.4pt, solid] coordinates {
        (0.04, 0.00003)(0.05, 0.00042)(0.06, 0.00742)(0.07, 0.06107)
        (0.08, 0.22696)(0.09, 0.53896)(0.10, 0.83511)
    };
    \addlegendentry{$[[730,18, 20]]$ UB}

    \addplot[color=colorPair2, mark=square, mark size=3pt,
             mark options={fill=white, draw=colorPair2, line width=0.8pt},
             line width=1.4pt, dashed] coordinates {
        (0.04, 0.00001)(0.05, 0.00032)(0.06, 0.00602)(0.07, 0.04611)
        (0.08, 0.20770)(0.09, 0.49026)(0.10, 0.80319)
    };
    \addlegendentry{$[[756,16, \leq 34]]$ BB\cite{Bravyi_2024}}

    \addplot[color=colorPair3, mark=triangle*, mark size=3.5pt,
             line width=1.4pt, solid] coordinates {
        (0.03, 0.00005)(0.04, 0.00052)(0.05, 0.00443)(0.06, 0.02945)
        (0.07, 0.11009)(0.08, 0.24355)(0.09, 0.47785)(0.10, 0.73853)
    };
    \addlegendentry{$[[372,14, 12]]$ UB}

    \addplot[color=colorPair3, mark=triangle, mark size=3.5pt,
             mark options={fill=white, draw=colorPair3, line width=0.8pt},
             line width=1.4pt, dashed] coordinates {
        (0.03, 0.00001)(0.04, 0.00023)(0.05, 0.00354)(0.06, 0.02945)
        (0.07, 0.09895)(0.08, 0.26129)(0.09, 0.49684)(0.10, 0.69266)
    };
    \addlegendentry{$[[360,12, \leq 24]]$ BB\cite{Bravyi_2024}}

    \end{loglogaxis}
    \end{tikzpicture}}%
    \par\smallskip
    \ref{ferlegendBB}
    \caption{Performance of selected UB codes compared with BB \cite{Bravyi_2024} codes of similar block lengths and rates.}
    \label{fig:fer_ub_bb}
    \vspace{-4.8mm}
\end{figure}

\definecolor{colorPairA}{HTML}{1F77B4}   
\definecolor{colorPairB}{HTML}{D62728}   
\definecolor{colorPairC}{HTML}{2CA02C}   

\begin{figure}[t]
    \centering
    \resizebox{0.79\columnwidth}{!}{%
    \begin{tikzpicture}
    \begin{loglogaxis}[
        scale only axis,
        width  = 6.5cm,
        height = 6.5cm,
        xlabel = {Physical Error Rate $p_x$},
        ylabel = {Logical Error Rate},
        xlabel style     = {font=\small},
        ylabel style     = {font=\small},
        tick label style = {font=\small},
        xmin = 0.018, xmax = 0.112,
        ymin = 5e-6,  ymax = 1.5,
        xtick       = {0.02, 0.03, 0.04, 0.05, 0.06, 0.07, 0.08, 0.09, 0.10},
        xticklabels = {$0.02$, $0.03$, $0.04$, $0.05$, $0.06$,
                       $0.07$, $0.08$, $0.09$, $0.10$},
        xticklabel style = {rotate=45, anchor=north east, font=\small},
        xtick align      = inside,
        ytick align      = inside,
        enlarge x limits = 0.02,
        enlarge y limits = 0.02,
        clip             = true,
        grid             = both,
        grid style       = {line width=0.4pt, black!25},
        minor grid style = {line width=0.3pt, black!15},
        minor tick num   = 1,
        axis line style  = {line width=0.8pt},
        legend style = {
            font         = \footnotesize,
            fill         = white,
            draw         = black,
            cells        = {anchor=west},
            column sep   = 8pt,
            row sep      = 1pt,
            inner sep    = 4pt,
            legend columns = 2,
        },
        legend to name    = ferlegendA,
        legend cell align = left,
    ]

    \addplot[color=colorPairA, mark=*, mark size=3pt,
             line width=1.4pt, solid] coordinates {
        (0.02, 0.00004)(0.03, 0.00121)(0.04, 0.01259)(0.05, 0.05737)
        (0.06, 0.16444)(0.07, 0.30754)(0.08, 0.50482)(0.09, 0.64490)(0.10, 0.83333)
    };
    \addlegendentry{$[[178,24,13]]$ UB}

    \addplot[color=colorPairA, mark=o, mark size=3pt,
             mark options={fill=white, draw=colorPairA, line width=0.8pt},
             line width=1.4pt, dashed] coordinates {
        (0.02, 0.00002)(0.03, 0.00082)(0.04, 0.00834)(0.05, 0.05889)
        (0.06, 0.15519)(0.07, 0.31365)(0.08, 0.48553)(0.09, 0.61633)(0.10, 0.78646)
    };
    \addlegendentry{$[[180,10, 15 \leq d \leq 18]]$ A5\cite{Panteleev_2021}}

    \addplot[color=colorPairB, mark=square*, mark size=3pt,
             line width=1.4pt, solid] coordinates {
        (0.02, 0.00001)(0.03, 0.00037)(0.04, 0.00629)(0.05, 0.03738)
        (0.06, 0.14205)(0.07, 0.27656)(0.08, 0.50842)(0.09, 0.69266)(0.10, 0.80749)
    };
    \addlegendentry{$[[234,26,14]]$ UB}

    \addplot[color=colorPairB, mark=square, mark size=3pt,
             mark options={fill=white, draw=colorPairB, line width=0.8pt},
             line width=1.4pt, dashed] coordinates {
        (0.02, 0.00003)(0.03, 0.00206)(0.04, 0.02835)(0.05, 0.11683)
        (0.06, 0.34337)(0.07, 0.52747)(0.08, 0.77104)(0.09, 0.89450)(0.10, 0.94118)
    };
    \addlegendentry{$[[254,28, 14 \leq d \leq 20]]$ A1\cite{Panteleev_2021}}

    \addplot[color=colorPairC, mark=triangle*, mark size=3.5pt,
             line width=1.4pt, solid] coordinates {
        (0.03, 0.00001)(0.04, 0.00015)(0.05, 0.00114)(0.06, 0.02124)
        (0.07, 0.17640)(0.08, 0.60204)(0.09, 0.88268)(0.10, 0.98052)
    };
    \addlegendentry{$[[1022,56, 21]]$ UB}

    \addplot[color=colorPairC, mark=triangle, mark size=3.5pt,
             mark options={fill=white, draw=colorPairC, line width=0.8pt},
             line width=1.4pt, dashed] coordinates {
        (0.03, 0.00017)(0.04, 0.00116)(0.05, 0.00769)(0.06, 0.05288)
        (0.07, 0.21495)(0.08, 0.51361)(0.09, 0.84358)(0.10, 0.98052)
    };
    \addlegendentry{$[[900,50,15]]$ A6\cite{Panteleev_2021}}

    \end{loglogaxis}
    
    \end{tikzpicture}}%
    \par\smallskip
    \ref{ferlegendA}
    \caption{Performance of selected UB codes compared with the GB codes A1, A5, and A6 \cite{Panteleev_2021}.}
    \label{fig:fer_ub_a}
    \vspace{-4.8mm}
\end{figure}
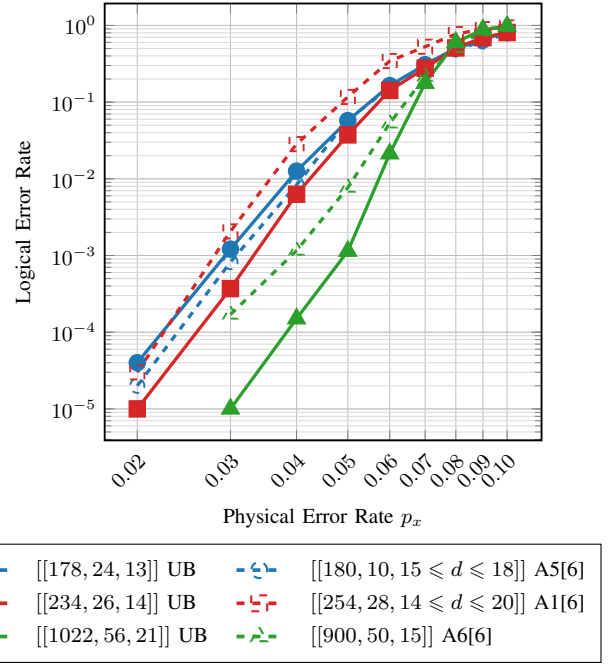
Table \ref{tab:poly-params} lists selected UB codes across different operating regimes, including short- and medium-length constructions, high-distance examples, and low-weight stabilizer configurations. Although the UB construction restricts the design space to a single-polynomial search, it still yields a broad range of competitive QLDPC codes with varying rates, distances, and block lengths. In particular, several UB codes with stabilizer weights $w=6$ and $w=8$ achieve parameters comparable to benchmark GB and BB constructions. The reported minimum distances are computed using the libraries in \cite{webster2026distance, codedistancepypi}.

Figures \ref{fig:fer_ub_bb} and \ref{fig:fer_ub_a} show the performance of selected UB codes under BP-OSD-0 decoding \cite{Panteleev_2021} over a binary symmetric channel with independent $X$-type errors. The simulations use serial normalized min-sum belief propagation with normalization factor $\alpha=0.875$ and maximum iteration count $I_{\max}=1000$. BP-OSD decoding is implemented using the libraries in \cite{Roffe_2020,Roffe_LDPC_Python_tools_2022}. For each physical error rate $p_x$, simulations are run until at least $150$ logical errors are collected.

Figure \ref{fig:fer_ub_bb} compares selected UB codes against representative BB codes \cite{Bravyi_2024} of comparable lengths and rates. All compared codes in this figure have stabilizer weight $w=6$. The results show that the proposed UB construction achieves competitive decoding performance despite being obtained from a substantially smaller single-polynomial search space. Moreover, in all three comparisons, the UB codes operate at slightly higher rates than their corresponding BB codes. In particular, the UB code $[[372,14,12]]$ achieves performance comparable to the BB code $[[360,12,\leq24]]$ while operating at a slightly higher rate.

Figure \ref{fig:fer_ub_a} compares UB codes with GB codes A1, A5, and A6 from \cite{Panteleev_2021}. The UB code $[[234,26,14]]$ consistently outperforms the A1 code $[[254,28,14\leq d\leq20]]$ across the simulated range while also using lower stabilizer weight ($w=8$ versus $w=10$). The UB code $[[1022,56,21]]$ also demonstrates better performance compared with the GB code A6 $[[900,50,15]]$ at almost the same rate. Overall, the results indicate that UB codes provide a favorable tradeoff between reduced search complexity, structural simplicity, and code performance, while remaining competitive or improving upon existing BB and GB constructions across various parameter regimes.

\section{Conclusion}
\label{sec:conclusion}
We introduced univariate bicycle (UB) codes, a structured subclass of
generalized bicycle (GB) codes with reduced design complexity and explicit
logical structure. Using this structure, we derived tractable upper bounds
on the minimum distance through low-weight logical representatives and
induced-cycle densities of associated circulant matrices. Numerical
results demonstrate competitive performance  despite the additional algebraic restriction. Several directions remain for future work, including tighter distance
analysis, extensions beyond the divisor setting, and improved decoding
methods tailored to UB codes. The explicit logical structure may also
enable systematic logical gate constructions and hardware-aware code
designs.

\bibliographystyle{IEEEtran}
{\footnotesize \bibliography{ref}}

@article{wang2022distance,
  title={Distance bounds for generalized bicycle codes},
  author={Wang, Renyu and Pryadko, Leonid P},
  journal={Symmetry},
  volume={14},
  number={7},
  pages={1348},
  year={2022},
  publisher={MDPI}
}

@article{gallager1962low,
  title={Low-density parity-check codes},
  author={Gallager, Robert},
  journal={IRE Transactions on information theory},
  volume={8},
  number={1},
  pages={21--28},
  year={1962},
  publisher={IEEE}
}

@article{7336474,
  title={Fifteen years of quantum {LDPC} coding and improved decoding strategies},
  author={Babar, Zunaira and Botsinis, Panagiotis and Alanis, Dimitrios and Ng, Soon Xin and Hanzo, Lajos},
  journal={iEEE Access},
  volume={3},
  pages={2492--2519},
  year={2015},
  publisher={IEEE}
}

@article{wang2024coprime,
  title={Coprime bivariate bicycle codes and their properties},
  author={Wang, Ming and Mueller, Frank},
  journal={arXiv e-prints},
  pages={arXiv--2408},
  year={2024}
}

@article{PhysRevA.88.012311,
  title={Quantum Kronecker sum-product low-density parity-check codes with finite rate},
  author={Kovalev, Alexey A and Pryadko, Leonid P},
  journal={Physical Review A—Atomic, Molecular, and Optical Physics},
  volume={88},
  number={1},
  pages={012311},
  year={2013},
  publisher={APS}
}

@article{Panteleev_2021,
   title={Degenerate Quantum {LDPC} Codes With Good Finite Length Performance},
   volume={5},
   ISSN={2521-327X},
   journal={Quantum},
   publisher={Verein zur Forderung des Open Access Publizierens in den Quantenwissenschaften},
   author={Panteleev, Pavel and Kalachev, Gleb},
   year={2021},
   month=nov, pages={585} }

@article{MacKay_2004,
  title={Sparse-graph codes for quantum error correction},
  author={MacKay, David JC and Mitchison, Graeme and McFadden, Paul L},
  journal={IEEE Transactions on Information Theory},
  volume={50},
  number={10},
  pages={2315--2330},
  year={2004},
  publisher={IEEE}
}

@article{Bravyi_2024,
  title={High-threshold and low-overhead fault-tolerant quantum memory},
  author={Bravyi, Sergey and Cross, Andrew W and Gambetta, Jay M and Maslov, Dmitri and Rall, Patrick and Yoder, Theodore J},
  journal={Nature},
  volume={627},
  number={8005},
  pages={778--782},
  year={2024},
  publisher={Nature Publishing Group UK London}
}

@article{postema2025existencecharacterisationbivariatebicycle,
  title={Existence and Characterisation of Bivariate Bicycle Codes},
  author={Postema, Jasper Johannes and Kokkelmans, Servaas JJMF},
  journal={arXiv preprint arXiv:2502.17052},
  year={2025}
}

@article{Roffe_2020,
  title={Decoding across the quantum low-density parity-check code landscape},
  author={Roffe, Joschka and White, David R and Burton, Simon and Campbell, Earl},
  journal={Physical Review Research},
  volume={2},
  number={4},
  pages={043423},
  year={2020},
  publisher={APS}
}

@article{Roffe_LDPC_Python_tools_2022,
  title={{LDPC}: Python tools for low density parity check codes},
  author={Roffe, Joschka},
  journal={PyPI https://pypi. org/project/ldpc},
  year={2022}
}

@article{Breuckmann_2021,
  title={Quantum low-density parity-check codes},
  author={Breuckmann, Nikolas P and Eberhardt, Jens Niklas},
  journal={PRX quantum},
  volume={2},
  number={4},
  pages={040101},
  year={2021},
  publisher={APS}
}

@article{vasic2025quantumlowdensityparitycheckcodes,
  title={Quantum Low-Density Parity-Check Codes},
  author={Vasic, Bane and Savin, Valentin and Pacenti, Michele and Borah, Shantom and Raveendran, Nithin},
  journal={arXiv preprint arXiv:2510.14090},
  year={2025}
}

@article{PhysRevA.54.1098,
  title = {Good quantum error-correcting codes exist},
  author = {Calderbank, A. R. and Shor, Peter W.},
  journal = {Phys. Rev. A},
  volume = {54},
  issue = {2},
  pages = {1098--1105},
  numpages = {0},
  year = {1996},
  month = {Aug},
  publisher = {American Physical Society},
  doi = {10.1103/PhysRevA.54.1098},
  url = {https://link.aps.org/doi/10.1103/PhysRevA.54.1098}
}

@article{PhysRevA.54.4741,
  title = {Simple quantum error-correcting codes},
  author = {Steane, A. M.},
  journal = {Phys. Rev. A},
  volume = {54},
  issue = {6},
  pages = {4741--4751},
  numpages = {0},
  year = {1996},
  month = {Dec},
  publisher = {American Physical Society},
  doi = {10.1103/PhysRevA.54.4741},
  url = {https://link.aps.org/doi/10.1103/PhysRevA.54.4741}
}

@ARTICLE{641542,
  author={Vardy, A.},
  journal={IEEE Transactions on Information Theory}, 
  title={The intractability of computing the minimum distance of a code}, 
  year={1997},
  volume={43},
  number={6},
  pages={1757-1766},
  doi={10.1109/18.641542}}

@ARTICLE{7891006,
  author={Dumer, Ilya and Kovalev, Alexey A. and Pryadko, Leonid P.},
  journal={IEEE Transactions on Information Theory}, 
  title={Distance Verification for Classical and Quantum LDPC Codes}, 
  year={2017},
  volume={63},
  number={7},
  pages={4675-4686},
  doi={10.1109/TIT.2017.2690381}}

@ARTICLE{10154053,
  author={Kapshikar, Upendra and Kundu, Srijita},
  journal={IEEE Transactions on Information Theory}, 
  title={On the Hardness of the Minimum Distance Problem of Quantum Codes}, 
  year={2023},
  volume={69},
  number={10},
  pages={6293-6302},
  doi={10.1109/TIT.2023.3286870}}

@misc{gottesman1997stabilizercodesquantumerror,
      title={Stabilizer Codes and Quantum Error Correction}, 
      author={Daniel Gottesman},
      year={1997},
      eprint={quant-ph/9705052},
      archivePrefix={arXiv},
      primaryClass={quant-ph},
      url={https://arxiv.org/abs/quant-ph/9705052}, 
}

@article{Wilde_2009,
   title={Logical operators of quantum codes},
   volume={79},
   ISSN={1094-1622},
   url={http://dx.doi.org/10.1103/PhysRevA.79.062322},
   DOI={10.1103/physreva.79.062322},
   number={6},
   journal={Physical Review A},
   publisher={American Physical Society (APS)},
   author={Wilde, Mark M.},
   year={2009},
   month=June }

@article{dastbasteh2025generalized,
  title={Generalized Bicycle Codes with Low Connectivity: Minimum Distance Bounds and Hook Errors},
  author={Dastbasteh, Reza and Larrarte, Olatz Sanz and Moncy, Arun John and Crespo, Pedro M and Martinez, Josu Etxezarreta and Otxoa, Ruben M},
  journal={arXiv preprint arXiv:2508.09082},
  year={2025}
}

@article{webster2026distance,
  title={Distance-Finding Algorithms for Quantum Codes and Circuits},
  author={Webster, Mark and Jacob, Abraham and Higgott, Oscar},
  journal={arXiv preprint arXiv:2603.22532},
  year={2026}
}

@misc{codedistancepypi,
  author       = {Michael Webster},
  title        = {{codeDistancePYPI}},
  howpublished = {\url{https://github.com/m-webster/codeDistancePYPI}},
  year         = {2025},
  note         = {Accessed: 2026-05-10}
}
\end{document}